\documentclass[journal,singleside,final]{IEEEtran}
\usepackage[nocompress]{cite}
\usepackage{amsmath}
\usepackage{amsthm}
\usepackage{amsfonts}
\usepackage{amssymb}
\usepackage[linesnumbered,ruled,vlined]{algorithm2e}
\usepackage{graphicx}
\usepackage{color}
\usepackage{xcolor}
\usepackage[hidelinks]{hyperref}
\usepackage[tight]{subfigure}
\usepackage{enumerate}
\usepackage{multirow}
\definecolor{orange}{rgb}{1,0.5,0}
\makeatletter
\newcommand{\removelatexerror}{\let\@latex@error\@gobble}
\makeatother

\SetKwInput{KwInput}{Input}                
\SetKwInput{KwOutput}{Output}              
\SetKw{Continue}{continue}
\SetKw{Break}{break}
\newtheorem{proposition}{Proposition}

\hypersetup{
    colorlinks=true,
    citecolor=blue,
    linkcolor=black,
    filecolor=magenta,      
    urlcolor=black,
    pdftitle={Overleaf Example},
    pdfpagemode=FullScreen,
    }

\usepackage[acronym,nonumberlist,style=super,toc]{glossaries}

\makeatletter
\newcommand*{\glsplainhyperlink}[2]{%
  \colorlet{currenttext}{.}
  \colorlet{currentlink}{\@linkcolor}
  \hypersetup{linkcolor=currenttext}
  \hyperlink{#1}{#2}%
  \hypersetup{linkcolor=currentlink}
}
\let\@glslink\glsplainhyperlink
\makeatother

\newacronym{5g}{5G}{fifith generation}
\newacronym{6g}{6G}{sixth generation}

\newacronym{awgn}{AWGN}{additive white Gaussian noise}
\newacronym{ao}{AO}{alternating optimization}

\newacronym{bs}{BS}{base station}
\newacronym{bdris}{BD-RIS}{beyond diagonal reconfigurable intelligent surface}

\newacronym{csi}{CSI}{channel state information}
\newacronym{cdf}{CDF}{cumulative  distribution  function}

\newacronym{dl}{DL}{downlink}
\newacronym{dris}{D-RIS}{diagonal reconfigurable intelligent surfaces}

\newacronym{gda}{GDA}{gradient descent algorithm}

\newacronym{iid}{i.i.d}{independent and identically distributed}
\newacronym{iot}{IoT}{Internet of Things}
\newacronym{irs}{IRS}{intelligent reflecting surfaces}
\newacronym{isac}{ISAC}{integrated sensing and communication}

\newacronym{jmld}{JMLD}{joint-multiuser maximum likelihood detector}

\newacronym{lut}{LUT}{look-up table}
\newacronym{lhs}{LHS}{left-hand side}

\newacronym{mpsk}{$M$-PSK}{M-ary phase shift keying}
\newacronym{mimo}{MIMO}{multiple-input-multiple-output}
\newacronym{miso}{MISO}{multiple-input-single-output}

\newacronym{mumiso}{MU-MISO}{multiuser multiple-input-single-output}

\newacronym{mld}{MLD}{maximum likelihood detector}
\newacronym{mrt}{MRT}{maximum ratio transmission}

\newacronym{ndris}{ND-RIS}{non-diagonal reconfigurable intelligent surfaces}

\newacronym{pc}{PC}{power control}

\newacronym{qos}{QoS}{quality of service}

\newacronym{ris}{RIS}{reconfigurable intelligent surface}
\newacronym{re}{RE}{reflecting element}
\newacronym{rf}{RF}{radio-frequency}
\newacronym{rhs}{RHS}{right-hand side}

\newacronym{sim}{SIM}{stacked intelligent metasurfaces}
\newacronym{s-param}{$S$-parameters}{scattering parameters}
\newacronym{sop}{SOP}{sum-of-product}

\newacronym{svd}{SVD}{singular value decomposition}

\newacronym{se}{SE}{spectral efficiency}
\newacronym{snr}{SNR}{signal to noise ratio}
\newacronym{siso}{SISO}{single-input-single-output}
\newacronym{sinr}{SINR}{signal to interference and noise ratio}
\newacronym{swipt}{SWIPT}{simultaneous wireless information and power transfer}

\newacronym{t-param}{$T$-parameters}{transfer scattering parameters}

\newacronym{ul}{UL}{uplink}
\newacronym{uav}{UAV}{unmanned aerial vehicle}

\newacronym{zf}{ZF}{zero forcing}
\newacronym{z-param}{$Z$-parameters}{impedance parameters}

\renewenvironment{thebibliography}[1]{
  \begin{oldthebibliography}{#1}
    \setlength{\itemsep}{0.01em}
    \setlength{\parskip}{-0.12em}
}
{
  \end{oldthebibliography}
}

\newtheorem{remark}{Remark}
\newtheorem{example}{Example}

\begin{document}
\title{\LARGE $T$-Parameters Based Modeling for Stacked Intelligent Metasurfaces: Tractable and  Physically Consistent Model}
\author{Hamad~Yahya,~
\IEEEmembership{Member,~IEEE,} Matteo~Nerini,~%
\IEEEmembership{Member,~IEEE,} Bruno~Clerckx,~%
\IEEEmembership{Fellow,~IEEE,} and~Merouane~Debbah,~%
\IEEEmembership{Fellow,~IEEE}
\thanks{Hamad Yahya is with the Department of Electrical Engineering, Khalifa University of Science and Technology, Abu Dhabi 127788, UAE (email: \href{mailto:hamad.myahya@ku.ac.ae}{hamad.myahya@ku.ac.ae}).

Matteo Nerini and Bruno Clerckx are with the  Department of Electrical and Electronic Engineering, Imperial College London, London SW7 2AZ, U.K., (e-mail: \{\href{mailto:m.nerini20@imperial.ac.uk}{m.nerini20}, \href{mailto:b.clerckx@imperial.ac.uk}{b.clerckx}\}@imperial.ac.uk).

Merouane Debbah is with the 6G Research Center, Khalifa University of Science and Technology, Abu Dhabi 127788, UAE (email: \href{mailto:merouane.debbah@ku.ac.ae}{merouane.debbah@ku.ac.ae}).
}}\maketitle

\begin{abstract}
This work develops a physically consistent model for \gls{sim} using multiport network theory and \gls{t-param}. Unlike the \gls{s-param} model, the developed \gls{t-param} model is simpler and more tractable. Moreover, the \gls{t-param} constraints for lossless reciprocal \glspl{ris} are derived. Additionally, a \gls{gda} is introduced to maximize sum-rate in \gls{sim}-aided multiuser scenarios, demonstrating that mutual coupling and feedback between consecutive layers enhance performance. However, increasing \gls{sim} layers with a fixed total number of elements typically degrades sum-rate, unless the simplified channel model employing Rayleigh-Sommerfeld diffraction coefficients is utilized.
\end{abstract}
\glsresetall

\markboth{Draft,~Vol.~xx, No.~xx,
Apr.~2025}{YAHYA \MakeLowercase{\textit{et al.}}: Tractable and Physically Consistent Channel Model for Stacked Intelligent Metasurfaces} 

\IEEEpeerreviewmaketitle
\glsresetall

\section{Introduction}\label{sec1}
\IEEEPARstart{T}{he} advent of intelligent surfaces such as \gls{ris} and diffractive neural networks gave rise to the new paradigm of \gls{sim} \cite{An2023-ICC,An2024-WCM,Rev1_1}, which is constructed by stacking layers of \gls{ris} in transmission mode to manipulate the electromagnetic waves. The novel idea of stacking multilayer communication devices dates back to 2012 when Boccia \textit{et al.} \cite{Boccia2012-TMTT} proposed stacking multilayer antenna-filter antennas system to achieve beam steering. In 2024, Dai \textit{et al.} \cite{Linglong2024-patent} have been granted a United States patent on multilayer \glspl{ris}. Recently, Nerini and Clerckx \cite{Matteo2024-LCOM} developed a physically consistent model for \gls{sim} using multiport network theory and cascaded \gls{s-param}, explicitly laying out the assumptions for the widely used simplified channel model \cite{An2023-ICC,An2024-JSAC}. However, their exact channel model is highly complex and non-tractable due to its nested nature and excessive number of matrix inversions. Similarly, Abrardo \textit{et al.} \cite{abrardo2025novel} proposed an \gls{z-param} model for \gls{sim}, constructing its equivalent \gls{z-param} matrix using band matrices. Therefore, we explore an alternative representation for \gls{sim} using \gls{t-param}, which conveniently characterize a chain of multiport networks through straightforward matrix multiplication of their individual \gls{t-param} matrices \cite{Markos2008wave}. This approach motivates the development of a new, physically consistent model for \gls{sim} that is equivalent to the cascaded \gls{s-param} model, yet it is less complex and more tractable. We further derive the constraints of \gls{t-param} for lossless and reciprocal \glspl{ris}, which are used to design a \gls{gda} that controls phase shifts of \gls{sim} to maximize the sum-rate in multiuser scenarios. Interested readers can find the source code for this article on \href{https://github.com/hmjasmi/sim-t-param}{https://github.com/hmjasmi/sim-t-param}.

\section{Cascaded Multiport Networks Modeling}\label{sec-system-model}
\begin{figure*}[t]
    \centering
    {\subfigure{\includegraphics[scale=0.7]{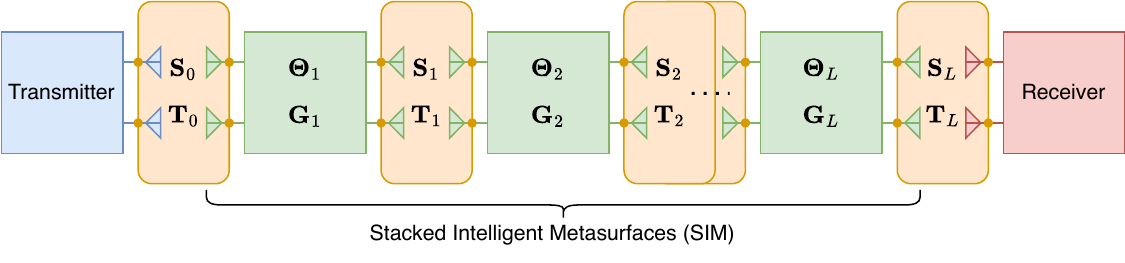}}}\vspace{-5mm}
    \caption{Multiport network model for \gls{sim}-aided communication system based on \gls{s-param} and \gls{t-param}, where $L$-layer \gls{sim} is considered.}
    \label{fig:BlockDiagram}\vspace{-5mm}
\end{figure*}

\subsection{S-parameters Formulation}\label{subsec-S-param}
The \gls{s-param} matrix of a multiport network associates the reflected waves with the incident waves via a linear transformation. Without loss of generality, we consider a cascade configuration in which a multiport network has ports on two opposite sides. The ports on the left side are the input ports, while the ports on the right side are output ports. When the number of input and output ports is equal, the multiport network is known as a ``balanced" multiport network \cite{Frei2008-TMTT}. Therefore, a balanced $2N$-port network can model \gls{sim} with $L$ cascaded layers of \glspl{ris} and $N$ cells each, where the $1$st layer is the input layer and the $L$th layer is the output layer \cite{An2023-ICC}. Hence, the reflected waves vector of \gls{sim} can be expressed as $\mathbf{b}_I=\mathbf{S}_I\mathbf{a}_I$, and expanded as follows
\begin{align}
\left[ 
\begin{array}{c}
\mathbf{b}_{\text{i}1} \\ 
\mathbf{b}_{\text{o}L}%
\end{array}%
\right] & =\mathbf{S}_I\left[ 
\begin{array}{c}
\mathbf{a}_{\text{i}1} \\ 
\mathbf{a}_{\text{o}L}%
\end{array}%
\right] =\left[ 
\begin{array}{c}
\mathbf{S}_{I,11}\mathbf{a}_{\text{i}1}+\mathbf{S}_{I,12}\mathbf{a}_{\text{o}L} \\ 
\mathbf{S}_{I,21}\mathbf{a}_{\text{i}1}+\mathbf{S}_{I,22}\mathbf{a}_{\text{o}L}%
\end{array}%
\right]   \label{eq-S}
\end{align}%
where $\mathbf{S}_I\in \mathbb{C}^{2N\times 2N}$ denotes the \gls{s-param} matrix of \gls{sim} with the ($i,j$)th block denoted as $\mathbf{S}_{I,ij}\in\mathbb{C}^{N\times N}$, $i,j\in\{1,2\}$, $\mathbf{b}_I=\left[ \mathbf{b}_{\text{i}1}^{T},\mathbf{b}_{\text{o%
}L}^{T}\right] ^{T}\in \mathbb{C}^{2N\times 1}$ denotes the reflected waves, $%
\mathbf{b}_{\text{i}1}\in \mathbb{C}^{N\times 1}$ and $\mathbf{b}_{\text{o}L%
}\in \mathbb{C}^{N\times 1}$ represent the reflected waves at the input ports of the input layer and the output ports of the output layer, $\mathbf{a}_{I}=\left[ \mathbf{a}_{\text{i}1}^{T},\mathbf{a}_{\text{o}L}^{T}%
\right] ^{T}\in \mathbb{C}^{2N\times 1}$ denotes the incident waves, $%
\mathbf{a}_{\text{i}1}\in \mathbb{C}^{N\times 1}$ and $\mathbf{a}_{\text{o}L%
}\in \mathbb{C}^{N\times 1}$ represent the incident waves at the input ports of the input layer and the output ports of the output layer.

The recursive approach presented in \cite{Matteo2024-LCOM} can be used to compute $\mathbf{S}_{I}$ such that 
\begin{equation}
\mathbf{S}_{I}=\mathcal{S}\left( \mathcal{S}\left( \cdots \mathcal{S}\left( 
\mathcal{S}\left( \boldsymbol{\Theta }_{1},\mathbf{S}_{1}\right) ,%
\boldsymbol{\Theta }_{2}\right) ,\ldots ,\mathbf{S}_{L-1}\right) ,%
\boldsymbol{\Theta }_{L}\right) \label{eq-S-casc}
\end{equation}%
where  $\mathcal{S}\left( \cdot \right) $ denotes a function that will be defined later, the \gls{s-param} matrix of the $l$th layer \gls{ris}, $\forall l\in \left\{ 1,\ldots ,L\right\} $, is denoted as $\boldsymbol{%
\Theta }_{l}=\left[ 
\begin{array}{cc}
\boldsymbol{\Theta }_{l,11} & \boldsymbol{\Theta }_{l,12} \\ 
\boldsymbol{\Theta }_{l,21} & \boldsymbol{\Theta }_{l,22}
\end{array}%
\right]\in \mathbb{C}^{2N\times 2N}$ and $\boldsymbol{\Theta }_{l,ij}\in\mathbb{C}^{N\times N}$. In addition, the \gls{s-param} of the transmission medium between the $l$th and $(l+1)$th \gls{ris} layer is denoted as $\mathbf{S}_{l}=\left[ 
\begin{array}{cc}
\mathbf{S}_{l,11} & \mathbf{S}_{l,12} \\ 
\mathbf{S}_{l,21} & \mathbf{S}_{l,22}%
\end{array}%
\right] \in \mathbb{C}^{2N\times 2N}$, $\forall l\in \left\{ 1,\ldots ,L-1\right\} $ , where $\mathbf{S}_{l,ij}\in \mathbb{C}^{N\times N}$, $\mathbf{S}_{l,11}$ and $\mathbf{S}_{l,22}$ refer to the antenna mutual coupling and self-impedance at the $l$th and $(l+1)$th \gls{ris} layer, $\mathbf{S}_{l,12}$ and $\mathbf{S}_{l,21}$ refer to the transmission \gls{s-param} matrices from the $(l+1)$th to $l$th \gls{ris} layer, and from the $l$th to $(l+1)$th \gls{ris} layer. The matrix $\mathbf{S}_{l}$ can be computed as 
$\mathbf{S}_{l}=\left(\mathbf{Z}_l+Z_0\mathbf{I}_{2N}\right)^{-1}\left(\mathbf{Z}_l-Z_0\mathbf{I}_{2N}\right)$,
where $\mathbf{I}_{2N}\in\mathbb{R}^{2N\times 2N}$ is the identity matrix, $Z_0=50$ $\Omega$ is the characteristic impedance, $\mathbf{Z}_l=\left[ 
\begin{array}{cc}
\mathbf{Z}_{l,11} & \mathbf{Z}_{l,12} \\ 
\mathbf{Z}_{l,21} & \mathbf{Z}_{l,22}%
\end{array}%
\right] \in \mathbb{C}^{2N\times 2N}$ is the transmission medium impedance matrix, $\mathbf{Z}_{l,ij}\in \mathbb{C}^{N\times N}$, $\mathbf{Z}_{l,11}$ and $\mathbf{Z}_{l,22}$ represent the impedance matrices at $l$th and $(l+1)$th \gls{ris} layers such that the diagonal entries are the self-impedance and the off-diagonal entries are the antenna mutual coupling, $\mathbf{Z}_{l,12}$ and $\mathbf{Z}_{l,21}$ refer to the transmission impedance matrices from the $(l+1)$th to $l$th \gls{ris} layer, and from the $l$th to $(l+1)$th \gls{ris} layer. In a reciprocal system, we have $\mathbf{Z}_{l,12}=\mathbf{Z}_{l,21}^T$. Fig. \ref{fig:BlockDiagram} shows the block diagram of \gls{sim}-aided communication system considering multiport network theory. 

Furthermore, introduced in \eqref{eq-S-casc},
$\mathcal{S}\left( \cdot \right) $ computes the equivalent \gls{s-param}
matrix for a cascade of networks. The process for cascading network $P$ and
network $Q$ which are balanced $2N$-port networks results in a balanced $2N$-port network $R$ with \gls{s-param} matrix $\mathbf{R}\in \mathbb{C}^{2N\times 2N}$
that is defined as $\mathbf{R} =\left[ 
\begin{array}{cc}
\mathbf{R}_{11} & \mathbf{R}_{12} \\ 
\mathbf{R}_{21} & \mathbf{R}_{22}%
\end{array}%
\right] \triangleq\mathcal{S}\left( \mathbf{P},\mathbf{Q}\right)$, and its ($i,j$)th blocks are defined in \cite[Eq. (12)--(16)]{Matteo2024-LCOM} such that
\begin{multline}
\left[\!\! 
\begin{array}{cc}
\mathbf{R}_{11} \\
\mathbf{R}_{12} \\ 
\mathbf{R}_{21} \\
\mathbf{R}_{22}
\end{array}\!\!
\right]=\left[\!\! 
\begin{array}{l}
\mathbf{P}_{11}+\mathbf{P}_{12}\left( \mathbf{I}_N-\mathbf{Q}_{11}\mathbf{P}%
_{22}\right) ^{-1}\mathbf{Q}_{11}\mathbf{P}_{21} \\ \mathbf{P}_{12}\left( 
\mathbf{I}_N-\mathbf{Q}_{11}\mathbf{P}_{22}\right) ^{-1}\mathbf{Q}_{12} \\ 
\mathbf{Q}_{21}\left( \mathbf{I}_N-\mathbf{P}_{22}\mathbf{Q}_{11}\right) ^{-1}\mathbf{P}_{21} \\ \mathbf{Q}_{22}+\mathbf{Q}_{21}\left( \mathbf{I}_N-\mathbf{P}%
_{22}\mathbf{Q}_{11}\right) ^{-1}\mathbf{P}_{22}\mathbf{Q}_{12}
\end{array}\!\!
\right] \label{eq-S-casc-oper}
\end{multline}
where $\mathbf{%
P},\mathbf{%
Q}\in \mathbb{C}^{2N\times 2N}$ are the \gls{s-param} matrices of networks $P$ and $Q$ with $\mathbf{P}_{ij},\mathbf{Q}_{ij}\in\mathbb{C}^{N\times N}$.

For a \gls{sim} implemented with lossless and reciprocal \glspl{ris}, the unitary constraint and symmetry constraints are imposed on $\boldsymbol{\Theta }_{l}$ such that
\begin{align}
\boldsymbol{\Theta }_{l}^{H}\boldsymbol{\Theta }_{l}&=\mathbf{I}_{2N}
\label{eq-losslessS}\\
\boldsymbol{\Theta }_{l}&=\boldsymbol{\Theta }_{l}^{T}.  \label{eq-SymmS}
\end{align}

To accurately model \gls{sim} using physically consistent models and capture the impact of deploying \gls{sim} in a radio environment, we consider the following assumptions:
\begin{enumerate}[{A}.1:]
    \item Unilateral approximation between transmitter and the $1$st layer \gls{ris}, and between the $L$th layer \gls{ris} and receiver.
    \item Matched transmitter and receiver antennas with $Z_0$.
    \item No mutual coupling between transmitter antennas, between receiver antennas, and between the $1$st layer \gls{ris} elements facing the transmitter.
\end{enumerate}
After expanding the general channel model in \cite[Eq. (35)]{Matteo2024-LCOM} and following the approach in \cite[Sec. V-C]{MatteoUniv2024-TWC}, the channel model based on the $S$-parameters can be written as
\begin{equation}
\mathbf{H}=\mathbf{H}_{RI}\mathbf{S}_{I,21}\mathbf{H}_{IT}\label{eq-channel-S}
\end{equation}%
where $\mathbf{H}_{RI}=\mathbf{S}_{L,21}\in \mathbb{C}^{{K\times N}}$, $\mathbf{S}_{L,21}$ is the  (2,1)th block of $\mathbf{S}_{L}\in\mathbb{C}^{(K+N)\times(K+N)}$ which is the \gls{s-param} matrix of the wireless channel between the $L$th layer \gls{ris} and the $K$ users equipped with single antennas , $\mathbf{H}_{IT}=\mathbf{S}_{0,21}\in 
\mathbb{C}^{{N\times K}}$, $\mathbf{S}_{0,21}$ is the (2,1)th block of $\mathbf{S}_{0}\in\mathbb{C}^{(K+N)\times(K+N)}$ which is the \gls{s-param} matrix of the wireless channel between the $K$ transmitter antennas and the $1$st layer \gls{ris}.

\begin{remark}
To compute the channel model in \eqref{eq-channel-S}, $\mathbf{S}_I$ needs to be computed recursively using \eqref{eq-S-casc}. Hence, the non-linear operator $\mathcal{S}\left( \cdot \right) $ has to be used $2\left( L-1\right)$ times. Such an expression is highly non-tractable due to its nested nature and block-level operations. In addition, the complexity associated with evaluating \eqref{eq-channel-S} is dominated by the number of matrix inversions, which grows excessively with $L$. Moreover, the optimization variables are highly coupled and difficult to optimize. 
To illustrate the disadvantages of the channel
model in \eqref{eq-channel-S}, we consider the following example.
\end{remark}

\begin{example}
Let $L=2$. Therefore, 
$\mathbf{S}_{I}=\mathcal{S}\left( \boldsymbol{\Theta }_{1},\mathcal{S}\left( 
\mathbf{S}_{1},\boldsymbol{\Theta }_{2}\right) \right)$ and its (2,1)th block can be expanded as
\begin{multline}
\mathbf{S}_{I,21} =  
\boldsymbol{\Theta }_{2,21} \left( \mathbf{I}_N - \mathbf{S}_{1,22} \boldsymbol{\Theta }_{2,11} \right)^{-1} \mathbf{S}_{1,21} \\
\times \left( \mathbf{I}_N - \boldsymbol{\Theta }_{1,22} \left( \mathbf{S}_{1,11} + \mathbf{S}_{1,12} \left( \mathbf{I}_N - \boldsymbol{\Theta }_{2,11} \mathbf{S}_{1,22} \right)^{-1} \right.\right. \\
\times \boldsymbol{\Theta }_{2,11} \mathbf{S}_{1,21} \Bigr) \Bigr)^{-1} \boldsymbol{\Theta }_{1,21}
\end{multline}%
where the number of matrix inversions in this expression is $N_{\text{inv-S}}=3$ and it grows
excessively with $L$. For
$L=3,4,5,6$, $N_{\text{inv-S}}=11,30,67,145$, which can approximated using a fourth order polynomial, i.e.,  $N_{\text{inv-S}}\approx\frac{2}{3}L^4-8\frac{1}{6}L^3+42\frac{1}{3}L^{2}-91\frac{5}{6}L+72$. We note that nested matrix inversions make controlling
optimization variables more complicated. Therefore, we propose a more
tractable formulation that is less complex and easier to optimize.
\end{example}

\subsection{T-parameters Formulation}\label{subsec-T-param}

The \gls{t-param} matrix of a multiport network have been introduced as a convenient mathematical concept for the cascade configuration, where it associates the reflected and incident waves at the input ports with the reflected and incident waves at the output ports via a linear
transformation. This transformation can be expressed for \gls{sim} as follows 
\begin{equation}
\left[ 
\begin{array}{c}
\mathbf{b}_{\text{i}1} \\ 
\mathbf{a}_{\text{i}1}%
\end{array}%
\right] =\mathbf{T}_{I}\left[ 
\begin{array}{c}
\mathbf{a}_{\text{o}L} \\ 
\mathbf{b}_{\text{o}L}%
\end{array}%
\right] =\left[ 
\begin{array}{c}
\mathbf{T}_{I,11}\mathbf{a}_{\text{o}L}+\mathbf{T}_{I,12}\mathbf{b}_{\text{o}L} \\ 
\mathbf{T}_{I,21}\mathbf{a}_{\text{o}L}+\mathbf{T}_{I,22}\mathbf{b}_{\text{o}L}%
\end{array}%
\right]   \label{eq-T}
\end{equation}%
where $\mathbf{T}_{I}\in \mathbb{C}^{2N\times 2N}$ denotes the \gls{t-param} matrix of \gls{sim} with the ($i,j$)th block denoted as $\mathbf{T}_{I,ij}\in\mathbb{C}^{N\times N}$. Furthermore, $\mathbf{T}_{I}$ can be
computed as $\mathbf{T}_{I}\triangleq \mathcal{T}\left( 
\mathbf{S}_{I}\right) $ and vice versa $\mathbf{S}_I$ can be computed as $\mathbf{S}_{I}\triangleq \mathcal{%
T}^{-1}\left( \mathbf{T}_{I}\right) $, where the definitions of $\mathbf{T}_{I,ij},\mathbf{S}_{I,ij}$ are given as \cite[Eq. (11)--(12)]{Frei2008-TMTT}, 
\begin{align}
\!\!\left[ 
\begin{array}{cc}
\mathbf{T}_{I_{11}} & \mathbf{T}_{I_{12}} \\ 
\mathbf{T}_{I_{21}} & \mathbf{T}_{I_{22}}%
\end{array}%
\right]\!&=\!\left[ 
\begin{array}{cc}
\mathbf{S}_{I_{12}}-\mathbf{S}_{I_{11}}\mathbf{S}_{I_{21}}^{-1}\mathbf{S}%
_{I_{22}} & \mathbf{S}_{I_{11}}\mathbf{S}_{I_{21}}^{-1} \\ 
-\mathbf{S}_{I_{21}}^{-1}\mathbf{S}_{I_{22}} & \mathbf{S}_{I_{21}}^{-1}%
\end{array}%
\right]\\
\!\!\left[ 
\begin{array}{cc}
\mathbf{S}_{I_{11}} & \mathbf{S}_{I_{12}} \\ 
\mathbf{S}_{I_{21}} & \mathbf{S}_{I_{22}}%
\end{array}%
\right]\!&=\!  \left[ 
\begin{array}{cc}
\!\!\mathbf{T}_{I_{12}}\mathbf{T}_{I_{22}}^{-1} & \mathbf{T}_{I_{11}}\!\!-\!\mathbf{T}%
_{I_{12}}\mathbf{T}_{I_{22}}^{-1}\mathbf{T}_{I_{21}}\!\!\!\! \\ 
\mathbf{T}_{I_{22}}^{-1} & -\mathbf{T}_{I_{22}}^{-1}\mathbf{T}_{I_{21}}
\end{array}
\right] .
\end{align}

While \gls{sim} layers are generally balanced multiport networks, generalizing the \gls{t-param} for devices with unbalanced multiport network would result in information loss \cite{Frei2008-TMTT}.
\begin{proposition}
The \gls{t-param} matrix of \gls{sim} with $L$-layer \glspl{ris} and $N$ cells can be expressed as
\begin{equation}
\mathbf{T}_{I}={\mathbf{G}}_{1}\mathbf{T}_{1}{\mathbf{G}}_{2}\cdots \mathbf{T%
}_{L-1}{\mathbf{G}}_{L}\label{eq-T-casc}
\end{equation}%
where ${\mathbf{G}}_{l}\in \mathbb{C}^{2N\times 2N}$ denotes the \gls{t-param} matrix of the $l$th layer \gls{ris}, $\forall l\in \left\{ 1,\ldots ,L\right\} $, and is given as ${\mathbf{G}}_{l}=\left[ 
\begin{array}{cc}
{\mathbf{G}}_{l,11} & {\mathbf{G}}_{l,12} \\ 
{\mathbf{G}}_{l,21} & {\mathbf{G}}_{l,22}%
\end{array}%
\right] \triangleq \mathcal{T}\left( \boldsymbol{\Theta }_{l}\right) $, with the ($i,j$)th block denoted as $\mathbf{G}_{I,ij}\in\mathbb{C}^{N\times N}$. In addition, the \gls{t-param} matrix of the transmission medium between the $l$th and $(l+1)$th \gls{ris} layer is denoted as $\mathbf{T}_{l}\in \mathbb{C}^{2N\times 2N}$ $\forall l\in \left\{ 1,\ldots ,L-1\right\} $, and is given as $\mathbf{T}_{l}=\left[ 
\begin{array}{cc}
\mathbf{T}_{l,11} & \mathbf{T}_{l,12} \\ 
\mathbf{T}_{l,21} & \mathbf{T}_{l,22}%
\end{array}%
\right] \triangleq \mathcal{T}\left( \mathbf{S}_{l}\right) $, $\forall l\in
\left\{ 1,2,\ldots ,L-1\right\} $, with the ($i,j$)th block denoted as $\mathbf{T}_{l,ij}\in\mathbb{C}^{N\times N}$. 
\end{proposition}
\begin{proof}
    Please see Appendix \ref{appA} for the proof. 
\end{proof}

Moreover, the constraints on ${\mathbf{G}}_{l}$ are given in the following proposition. 
\begin{proposition}
For a \gls{sim} implemented with lossless \glspl{ris}, the pseudo-unitary constraint is imposed on $\mathbf{G}_l$ such that
\begin{equation}
{\mathbf{G}}_{l}^{H}\boldsymbol{\Sigma}_{2N}{\mathbf{G}}_{l}=\boldsymbol{\Sigma }_{2N}
\label{eq-pseudoUnitaryT}
\end{equation}%
where $\boldsymbol{\Sigma }_{2N}=\mathrm{blkdiag}\left( \mathbf{I}_{N},-\mathbf{I}%
_{N}\right) $ and $\mathrm{blkdiag}(\cdot)$ returns a block diagonal matrix with the blocks defined in the argument. In addition, reciprocal \glspl{ris} impose the complex-conjugate persymmetric constraint on $\mathbf{G}_l$ such that \begin{equation}
{\mathbf{G}}_{l}=\mathbf{J}_{2N}{\mathbf{G}}_{l}^{\ast }\mathbf{J}_{2N}
\label{eq-blocksymmT}
\end{equation}%
where $\mathbf{J}_{2N}=\mathrm{blkantidiag}\left( \mathbf{I}_{N},\mathbf{I}%
_{N}\right)$ is known as the exchange matrix and $\mathrm{blkantidiag}\left(\cdot\right)$ returns a block antidiagonal matrix with the blocks defined in the argument.
\end{proposition}
\begin{proof}
    Please see Appendix \ref{appB} for the proof.
\end{proof}
The channel model based on the \gls{t-param} is given in the following proposition.
\begin{proposition}
Assuming A.1--A.3, the channel model based on the $T$%
-parameters is given as%
\begin{equation}
\mathbf{H}=\mathbf{H}_{RI}\mathbf{T}_{I,22}^{-1}\mathbf{H}_{IT}.\label{eq-channel-T}
\end{equation}%
\end{proposition}
\begin{proof}
    Please see Appendix \ref{appC} for the proof.
\end{proof}
\begin{remark}
To compute the channel model in \eqref{eq-channel-T}, $\mathbf{T}_I$ needs to be computed using \eqref{eq-T-casc} and its (2,2)th block should be extracted. Hence, it requires $2(L-1)$ matrix multiplications, which are linear operations, unlike the \gls{s-param} counterpart, which
requires $2\left( L-1\right) $ non-linear $\mathcal{S}(\cdot)$ operators. Also, \eqref{eq-T-casc} is more tractable due to its compact nature, unlike the \gls{s-param} counterpart which
is nested and block-level as seen in \eqref{eq-S-casc}. Moreover, the complexity associated with evaluating \eqref{eq-channel-T} is less since it requires a single matrix inversion, unlike the \gls{s-param} counterpart
which requires an excessive number of matrix inversions that grows excessively with $L$. To illustrate the advantages of the channel
model in \eqref{eq-channel-T}, we consider the following example.
\end{remark}

\begin{example}
Let $L=2$. Therefore, $
\mathbf{T}_{I}={\mathbf{G}}_{1}\mathbf{T}_{1}{\mathbf{G}}_{2}$, and its (2,2)th block can be expanded as 
\begin{multline}
\mathbf{T}_{I,22} =  
{\mathbf{G}_{1,21}\mathbf{T}_{1,11}\mathbf{G}_{2,12}+\mathbf{G}_{1,21}%
\mathbf{T}_{1,12}\mathbf{G}_{2,22}} \\ +{\mathbf{G}_{1,22}\mathbf{T}_{1,21}%
\mathbf{G}_{2,12}+\mathbf{G}_{1,22}\mathbf{T}_{1,22}\mathbf{G}_{2,22}}.
\end{multline}%
\end{example}

Table \ref{tab:s-t-param} summarizes the \gls{s-param} and \gls{t-param} modeling in terms of channel models, tractability, controllability,
complexity and constraints, which are important features to solve any optimization
problem involving \gls{sim}.
\begin{table}[]
    \centering
    \caption{\gls{s-param} and \gls{t-param} Modeling Summary.}
    \begin{tabular}{|l|c|c|}
    \hline
    & \gls{s-param} & \gls{t-param} \\ \hline
    Channel Model & \eqref{eq-channel-S} & \eqref{eq-channel-T} \\
    Tractability & Non-tractable & Tractable \\ 
    Controllability & Difficult & Easy \\ 
    Complexity & $\mathcal{O}(\frac{2}{3}L^4N^3)$ & $\mathcal{O}(N^3)$ \\ \hline\hline
    Lossless Constraint & Unitary \eqref{eq-losslessS} & Pseudo-unitary \eqref{eq-pseudoUnitaryT} \\
    Reciprocal Constraint & Symmetric \eqref{eq-SymmS} & Persymmetric \eqref{eq-blocksymmT} \\ \hline
\end{tabular}

    \label{tab:s-t-param}\vspace{-4mm}
\end{table}

\section{T-Parameters Design for SIM}\label{sec-design}
Consider \gls{sim}-aided communication system with a \gls{bs} consisting of $K$ \gls{rf} chains and \gls{sim} with $L$-layer \glspl{ris} and $N$ cells. The \gls{bs} serves $K$ single-antenna users. Therefore, the vector-form received signals at the users can be written as
\begin{equation}
    \mathbf{y} = \mathbf{HPx} + \mathbf{n}
\end{equation}
where $\mathbf{x}\in\mathbb{C}^{K\times1}$ denotes the transmitted signals, $\mathbf{n}\in\mathbb{C}^{K\times1}$ is the \gls{awgn} with $N_0$ the power spectral density,  $\mathbf{P}=\mathrm{diag}\left([\sqrt {p_1},\sqrt {p_2},\ldots,\sqrt {p_K}]^T\right)\in\mathbb{R}^{K\times K}$ denotes the power allocation matrix satisfying, $\sum_kp_k=P_{\max}$, where $P_{\max}$ is the \gls{bs} maximum power budget. 
Furthermore, the objective is to maximize the sum-rate, which can be achieved by considering a two-stage design \cite{Hamad2024-OJCOM}. In stage 1, the \glspl{ris} phase shifts are optimized assuming $\mathbf{P}$ is uniformly fixed, while stage 2 considers $\mathbf{P}$ design assuming fixed \glspl{ris} phase shifts. Therefore, the \gls{sinr} for the $k$th user can be defined as $
\gamma_k= \frac{p_k\left\vert \left[\mathbf{H}\right]_{k,k}\right\vert ^{2}}{\sum_{i\neq k}p_i\left\vert \left[\mathbf{H}\right]_{k,i}\right\vert ^{2}+N_0}$ \cite[Eq. (11)]{Hamad2024-OJCOM}, where $\mathbf{H}$ is defined in \eqref{eq-channel-T}. Consequently, when considering transmissive single-connected \glspl{ris}, we can formulate stage 1 sum-rate maximization as
\begin{subequations}
\begin{align}
    (P1): & \,\max_{\boldsymbol{\Phi}}f\left(\boldsymbol{\Phi}\right) \triangleq {\textstyle\sum_{k=1}^{K}}\log_2 \!\left(1+\gamma_k\right)\label{optA}
 \\
     \text{Subject to,  }& \mathbf{G}_{l} = \mathrm{blkdiag}\left(\overline{\boldsymbol{\Theta }}_{l},\overline{\boldsymbol{\Theta }}_{l}^{-1}\right) \label{optB}\\
     & \overline{\boldsymbol{\Theta }}_{l}=\mathrm{diag}\left(\left[\exp\!\left(j\phi_{l,1}\right),\ldots,\exp\!\left(j\phi_{l,N}\right)\right]\right) \label{optC}\\
    &\phi_{l,n}\in[0,2\pi),\, \forall l, n\label{optD}
    \end{align}
\end{subequations}
where \eqref{optA} is is the objective function, \eqref{optB}--\eqref{optD} are the reciprocity and lossless constraints on the \gls{t-param} matrix of \gls{sim} layers. Note that the \gls{t-param} constraints in \eqref{eq-pseudoUnitaryT}--\eqref{eq-blocksymmT} simplify to \eqref{optB}--\eqref{optD} because transmissive single-connected \glspl{ris} are considered. In addition, $\boldsymbol{\Phi}\in\mathbb{R}^{L\times N}$ is the design variable that essentially stores the \gls{sim} phase shifts such that $[\boldsymbol{\Phi}]_{l,n}=\phi_{l,n}$.

The problem in ($P1$) can be efficiently solved via a \gls{gda} which iteratively updates $\boldsymbol{\Phi}^{(t)}$
using the gradient of \eqref{optA}. The update rule follows is given as
$\boldsymbol{\Phi}^{(t+1)} = \boldsymbol{\Phi}^{(t)} + \alpha^{(t)} \nabla f(\boldsymbol{\Phi}^{(t)})$, 
where $\alpha^{(t)}$ is the Armijo step to ensure convergence \cite{An2023-ICC}, $t$ is the iteration index, and the gradient of the objective function is computed efficiently considering the numerical first order approximation. We initialize $\boldsymbol{\Phi}$ based on the simplified channel model \cite[Eq. (40)]{Matteo2024-LCOM}, and the maximum ratio transmission algorithm developed in \cite{Hamad2024-OJCOM}. In addition, a suitable power allocation shall be considered for stage 2.
\vspace{-3mm}
\section{Numerical Results}\label{sec-results}
This section presents the numerical results for the sum-rate considering the physically consistent channel model developed in Sec. \ref{subsec-T-param}. Unless otherwise stated, $\mathbf{Z}_l$ is computed based on \cite[Eq. (6)]{Gradoni2021-WCL}. The \gls{ris} elements are arranged uniformly over the $yz$-plane, where $N=N_yN_z$, $N_y$ and $N_z$ are the number of elements along $y$-axis and $z$-axis. The separation of adjacent elements in $y$-axis and $z$-axis are denoted by $l_y,l_z$, while the separation between adjacent \gls{sim} layers is denoted by $l_x$. The elements are assumed to be parallel to $z$-axis with identical lengths $\iota=\lambda/4$ and radii $\rho=\lambda/500$, where $\lambda$ denotes the wavelength of the $28$ GHz operating frequency. In addition, \gls{iid} Rayleigh fading channels with unit variance are considered for $\mathbf{H}_{RI}$ and $\mathbf{H}_{IT}$, with $K=2$. In the following, the sum-rate results are presented considering exact \eqref{eq-channel-T} and simplified channel models \cite[Eq. (40)]{Matteo2024-LCOM}. Specifically: 1) ``EE": optimization and evaluation are based on exact channel. 2) ``SE": optimization is based on simplified channel, and evaluation is based on exact channel. 3) ``SS": optimization and evaluation are based on simplified channel.

Fig. \ref{fig:res1}a illustrates the convergence of the \gls{gda}, where $L=3$ and $N=N_yN_z=6\times 6=36$. It is noted that accounting for the exact channel improves the sum-rate. Furthermore, while packing elements closer to each other makes mutual coupling effect stronger, the sum-rate improvement becomes more significant. Fig. \ref{fig:res1}b and Fig. \ref{fig:res1}c illustrate the sum-rate against the number of layers $L\in\{2,3,4,6\}$, where we choose $NL=72$, $N_y=6$, $N_z=N_y/N$, $l_x=\lambda/12(L-1)$, $l_y=\lambda/2$, $l_z=36\lambda/2N$ to relax the impact of mutual coupling and focus on the impact of transmission coefficients. Furthermore, hundred realizations are assumed for the Monte-Carlo simulations, $N_0=1$, $P_{\max}=2$ and uniform power allocation is considered. It is noted from Fig. \ref{fig:res1}b that ``EE", ``SE" and ``SS" have the same sum-rate for $L = 2$ because the simplified channel is only exact for this case. Furthermore, the sum-rate does not improve with $L$ as $N$ is reduced and the losses of the transmission medium are increased. However, the opposite is observed with Fig. \ref{fig:res1}c which considers the Rayleigh-Sommerfeld diffraction coefficients for $\mathbf{S}_{l,21}$ assuming square patch elements with $\iota=\lambda/4$ \cite[Eq. (1)]{An2023-ICC}. It is worth noting that Rayleigh-Sommerfeld diffraction coefficients only accurately models reality when 1) the \gls{ris} element surface area is $\gg\lambda$, 2) and the electromagnetic field is not observed very close to the surface \cite[Sec. III-B]{Ajam2022-TCOM}, \cite{Rev2_1}. Both conditions are breached in the literature and in this scenario. Consequently, adopting Rayleigh-Sommerfeld diffraction coefficients is questionable. 

\begin{figure}
    \centering
    {\subfigure[]{\includegraphics[scale=0.5,trim={3.5mm 2mm 6mm 0},clip]{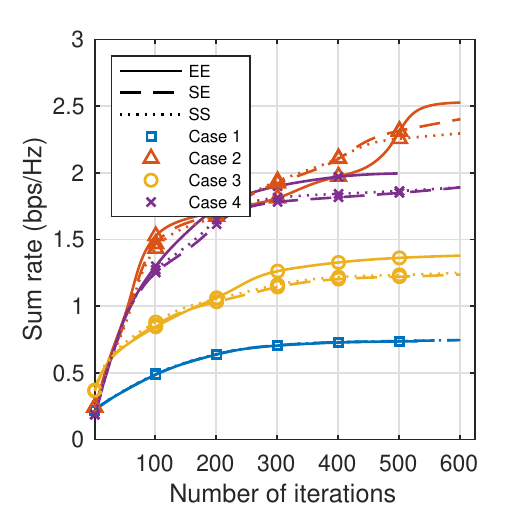}}}
    {\subfigure[]{\includegraphics[scale=0.5,trim={7.75mm 2mm 3.5mm 0},clip]{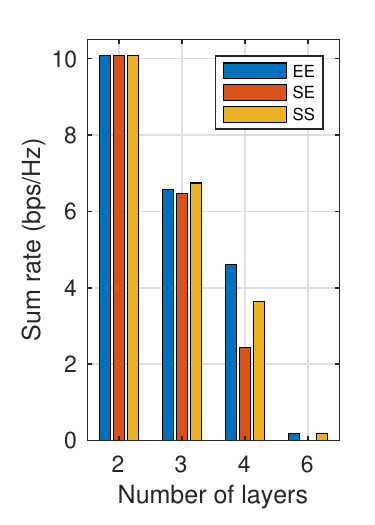}}}
    {\subfigure[]{\includegraphics[scale=0.5,trim={7.75mm 2mm 3.5mm 0},clip]{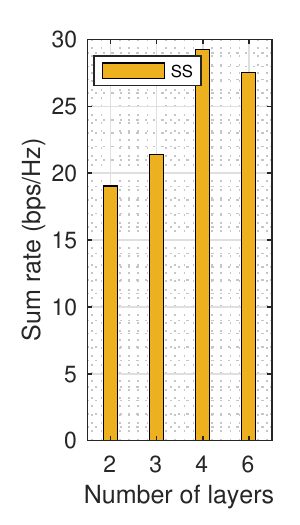}}}\vspace{-2mm}
    \caption{
    (a) Convergence of \gls{gda}. Case 1: ($l_x=l_y=l_z=\lambda/2$), Case 2: ($l_x=l_y=l_z=\lambda/3$), Case 3: ($l_x=\lambda/3,l_y=l_z=\lambda/2$), Case 4: ($l_x=\lambda/2,l_y=l_z=\lambda/3$).\\(b) and (c) Sum-rate against $L\in\{2,3,4,6\}$, $NL=72$, $N_y=6$, $N_z=N_y/N$, $l_x=\lambda/12(L-1)$, $l_y=\lambda/2$, $l_z=36\lambda/2N$.}\vspace{-5mm}
    \label{fig:res1}
\end{figure}

\vspace{-3mm}
\appendices
\section{Proof of T-parameters Matrix for SIM}\label{appA}
To prove \eqref{eq-T-casc}, consider the cascade of networks $P$ and $Q$, where network $P$ is the input layer and network $Q$ is output layer. Let networks $P$ and $Q$ be two $2N$-port balanced networks with \gls{t-param} matrices
denoted as $\mathbf{T}_p$ and $\mathbf{T}_q$. The incident waves at the input and output ports of network $P$ ($Q$) are denoted as $\mathbf{a}_{\text{i}p}, \mathbf{a}_{\text{o}p}$ ($\mathbf{a}_{\text{i}q}, \mathbf{a}_{\text{o}q}$), while reflected waves at the input and output ports of network $P$ ($Q$) are denoted as  $\mathbf{b}_{\text{i}p}, \mathbf{b}_{\text{o}p}$ ($\mathbf{b}_{\text{i}q}, \mathbf{b}_{\text{o}q}$). Following the definition of \gls{t-param} shown in \eqref{eq-T}, we can write 
\begin{align}
\left[ 
\begin{array}{c}
\mathbf{b}_{\text{i}p} \\ 
\mathbf{a}_{\text{i}p}%
\end{array}%
\right] &=\mathbf{T}_p\left[ 
\begin{array}{c}
\mathbf{a}_{\text{o}p} \\ 
\mathbf{b}_{\text{o}p}%
\end{array}\label{eq-proofA-1}
\right] \\
\left[ 
\begin{array}{c}
\mathbf{b}_{\text{i}q } \\ 
\mathbf{a}_{\text{i}q }%
\end{array}%
\right] &=\mathbf{T}_q\left[ 
\begin{array}{c}
\mathbf{a}_{\text{o}q } \\ 
\mathbf{b}_{\text{o}q }%
\end{array}%
\right] \label{eq-proofA-2}
\end{align}%
and because of the cascade configuration, we have $\mathbf{b}_{\text{i}q }=\mathbf{a}_{\text{o}p}$ and $%
\mathbf{a}_{\text{i}q }=\mathbf{b}_{\text{o}p}$. Hence, we can write \eqref{eq-proofA-2} as 
\begin{equation}
\left[ 
\begin{array}{c}
\mathbf{a}_{\text{o}p} \\ 
\mathbf{b}_{\text{o}p}%
\end{array}%
\right] =\mathbf{T}_q\left[ 
\begin{array}{c}
\mathbf{a}_{\text{o}q } \\ 
\mathbf{b}_{\text{o}q }%
\end{array}%
\right]. \label{eq-proofA-3}
\end{equation}%
Hence, by substituting \eqref{eq-proofA-3} in \eqref{eq-proofA-1} we obtain
\begin{align}
\left[ 
\begin{array}{c}
\mathbf{b}_{\text{i}p} \\ 
\mathbf{a}_{\text{i}p}%
\end{array}%
\right] & =\mathbf{T}_p\mathbf{T}_q\left[ 
\begin{array}{c}
\mathbf{a}_{\text{o}q } \\ 
\mathbf{b}_{\text{o}q }%
\end{array}%
\right]
\end{align}
Consequently, the equivalent multiport network has a \gls{t-param} matrix that can be denoted as $\mathbf{T}_{pq}\triangleq\mathbf{T}_p\mathbf{T}_q$. Following the same approach, \eqref{eq-T-casc} is proved. 
\vspace{-3mm}
\section{Proof of T-parameters Constraints}\label{appB}
To prove the lossless constraint imposed on $\mathbf{G}_{l}$, we
denote the incident waves (reflected waves) at the input and output ports of the $l$th layer \gls{ris} as $\mathbf{a}_{\text{i}l}$ and $\mathbf{a}_{\text{o}l}$ ($\mathbf{b}_{\text{i}l}$ and $\mathbf{b}_{\text{o}l}$). Consider the energy conservation law \cite{Markos2008wave}, we can write $\left\Vert \mathbf{b}_{\text{i}l}\right\Vert ^{2}-\left\Vert \mathbf{a}_{%
\text{i}l}\right\Vert ^{2}=\left\Vert \mathbf{a}_{\text{o}l}\right\Vert
^{2}-\left\Vert \mathbf{b}_{\text{o}l}\right\Vert ^{2}$. This can be represented in a vector form as follows, %
\begin{equation}
\left[ 
\begin{array}{c}
\mathbf{b}_{\text{i}l} \\ 
\mathbf{a}_{\text{i}l}%
\end{array}%
\right] ^{H}\boldsymbol{\Sigma }_{2N}\left[ 
\begin{array}{c}
\mathbf{b}_{\text{i}l} \\ 
\mathbf{a}_{\text{i}l}%
\end{array}%
\right] =\left[ 
\begin{array}{c}
\mathbf{a}_{\text{o}l} \\ 
\mathbf{b}_{\text{o}l}%
\end{array}%
\right] ^{H}\boldsymbol{\Sigma }_{2N}\left[ 
\begin{array}{c}
\mathbf{a}_{\text{o}l} \\ 
\mathbf{b}_{\text{o}l}%
\end{array}%
\right]. \label{eq-proofB-1}
\end{equation}%
Following the definition in \eqref{eq-T} with a notation change and applying the hermitian we can obtain
\begin{equation}
\left[ 
\begin{array}{c}
\mathbf{b}_{\text{i}l} \\ 
\mathbf{a}_{\text{i}l}%
\end{array}%
\right] ^{H}=\left[ 
\begin{array}{c}
\mathbf{a}_{\text{o}l} \\ 
\mathbf{b}_{\text{o}l}%
\end{array}%
\right] ^{H}\mathbf{G}_{l}^{H}.\label{eq-proofB-2}
\end{equation}%
By substituting \eqref{eq-proofB-2} in \eqref{eq-proofB-1}, we obtain 
\begin{equation}
\!\!\!\!\left[ 
\begin{array}{c}
\mathbf{a}_{\text{o}l} \\ 
\mathbf{b}_{\text{o}l}%
\end{array}%
\right] ^{H}\mathbf{G}_{l}^{H}\boldsymbol{\Sigma }_{2N}\mathbf{G}_{l}\!\left[ 
\begin{array}{c}
\mathbf{a}_{\text{o}l} \\ 
\mathbf{b}_{\text{o}l}%
\end{array}%
\right]  \!=\!\left[ 
\begin{array}{c}
\mathbf{a}_{\text{o}l} \\ 
\mathbf{b}_{\text{o}l}%
\end{array}%
\right] ^{H}\boldsymbol{\Sigma }_{2N}\! \left[ 
\begin{array}{c}
\mathbf{a}_{\text{o}l} \\ 
\mathbf{b}_{\text{o}l}%
\end{array}%
\right]
\end{equation}
which proves \eqref{eq-pseudoUnitaryT}.

A reciprocal \gls{ris} imposes the following on $\mathbf{G}_{l}$
\begin{align}
\left[ 
\begin{array}{c}
\mathbf{a}_{\text{i}l} \\ 
\mathbf{b}_{\text{i}l}%
\end{array}%
\right] ^{\ast }& =\mathbf{G}_l\left[ 
\begin{array}{c}
\mathbf{b}_{\text{o}l} \\ 
\mathbf{a}_{\text{o}l}%
\end{array}%
\right] ^{\ast }  \notag \\
\mathbf{J}_{2N}\left[ 
\begin{array}{c}
\mathbf{b}_{\text{i}l} \\ 
\mathbf{a}_{\text{i}l}%
\end{array}%
\right] ^{\ast }& =\mathbf{G}_l\mathbf{J}_{2N}\left[ 
\begin{array}{c}
\mathbf{a}_{\text{o}l} \\ 
\mathbf{b}_{\text{o}l}%
\end{array}%
\right] ^{\ast }.\label{eq-proofB-3}
\end{align}%
By taking the transpose of \eqref{eq-proofB-2}, we obtain
\begin{equation}
\left[ 
\begin{array}{c}
\mathbf{b}_{\text{i}l} \\ 
\mathbf{a}_{\text{i}l}%
\end{array}%
\right] ^{\ast }=\mathbf{G}_l^{\ast }\left[ 
\begin{array}{c}
\mathbf{a}_{\text{o}l} \\ 
\mathbf{b}_{\text{o}l}%
\end{array}%
\right] ^{\ast }\label{eq-proofB-4}
\end{equation}%
Hence, we substitute \eqref{eq-proofB-4} in \eqref{eq-proofB-3} to obtain
\begin{equation}
\mathbf{J}_{2N}\mathbf{G}_{l}^{\ast }\left[ 
\begin{array}{c}
\mathbf{a}_{\text{o}l} \\ 
\mathbf{b}_{\text{o}l}%
\end{array}%
\right] ^{\ast }=\mathbf{G}_{l}\mathbf{J}_{2N}\left[ 
\begin{array}{c}
\mathbf{a}_{\text{o}l} \\ 
\mathbf{b}_{\text{o}l}%
\end{array}%
\right] ^{\ast }.
\end{equation}%
Therefore, 
$\mathbf{J}_{2N}\mathbf{G}_{l}^{\ast }=\mathbf{G}_{l}\mathbf{J}_{2N}$,
which proves
\eqref{eq-blocksymmT}.
\vspace{-3mm}
\section{Proof of the SIM Channel Model Based on T-parameters}\label{appC} 
 The proof is given by assuming $K=N$, which does not have an impact on the preposition. Hence, the \gls{t-param} matrix of the wireless channel between the $L$th layer \gls{ris} and the users (transmitter antennas and
the $1$st layer \gls{ris}) can be denoted by $\mathbf{T}_{L}\in\mathbb{C}^{2N\times2N}$ ($\mathbf{T}_{0}\in\mathbb{C}^{2N\times2N}$).  Therefore, the equivalent \gls{t-param} matrix of the cascade can be written as $\mathbf{T}=\mathbf{T}_{0}\mathbf{T}_{I}\mathbf{T}_{L}$, where its ($i,j$)th block is denoted by $\mathbf{T}_{ij}$. By denoting the incident waves (reflected waves) at the transmitting antennas ports and the receiving antenna ports as $\mathbf{a}_{\text{i}}$ and $\mathbf{a}_{\text{o}}$ ($\mathbf{b}_{\text{i}}$ and $\mathbf{b}_{\text{o}}$), we can define the voltages vector at the transmitting antennas ports as%
\begin{equation}
\mathbf{v}_{\text{i}}=\mathbf{a}_{\text{i}}+\mathbf{b}_{\text{i}}  
=\left( \mathbf{T}_{21}+\mathbf{T}_{11}\right) \mathbf{a}_{\text{o}%
}+\left( \mathbf{T}_{22}+\mathbf{T}_{12}\right) \mathbf{b}_{\text{o}}
\end{equation}%
Therefore, we get an expression for $\mathbf{a}_{\text{o}}$ and $\mathbf{b}_{%
\text{o}}$ such that%
\begin{align}
\mathbf{a}_{\text{o}}&=\left( \mathbf{T}_{21}+\mathbf{T}_{11}\right)
^{-1}\left( \mathbf{v}_{\text{i}}-\left( \mathbf{T}_{22}+\mathbf{T}%
_{12}\right) \mathbf{b}_{\text{o}}\right) \\
\mathbf{b}_{\text{o}}&=\left( \mathbf{T}_{22}+\mathbf{T}_{12}\right)
^{-1}\left( \mathbf{v}_{\text{i}}-\left( \mathbf{T}_{21}+\mathbf{T}%
_{11}\right) \mathbf{a}_{\text{o}}\right) .
\end{align}%
Furthermore, we can express voltages vector at the receiving antennas ports as
\begin{multline}
\mathbf{v}_{\text{o}} =\mathbf{a}_{\text{o}}+\mathbf{b}_{\text{o}} 
 =\left( \mathbf{T}_{22}+\mathbf{T}_{12}\right) ^{-1}\mathbf{v}_{\text{i}%
}\\+\left( \mathbf{I}_N-\left( \mathbf{T}_{22}+\mathbf{T}_{12}\right)
^{-1}\left( \mathbf{T}_{21}+\mathbf{T}_{11}\right) \right) \mathbf{a}_{\text{%
o}}.
\end{multline}%
Assuming A.1--A.3 dictates $\mathbf{a}_{\text{o}}=\mathbf{0}$\cite[Eq. (33)]{Matteo2024-LCOM}, $\mathbf{T}_{0}=\mathrm{blkdiag}\left(\mathbf{0},\mathbf{T}_{0,22}\right)$ and $\mathbf{T}_{L}=\mathrm{blkdiag}\left(\mathbf{T}_{L,11},\mathbf{T}_{L,22}\right)$. Therefore, $\mathbf{T}_{11}=\mathbf{T}_{12}=
\mathbf{0}$, and we get $
\mathbf{v}_{\text{o}} =\mathbf{T}_{22} ^{-1}%
\mathbf{v}_{\text{i}}$. By noting that $\mathbf{v}_{\text{o}}=\mathbf{Hv}_{\text{i}}$, we can express the channel as 
\begin{align}
\mathbf{H} &=\mathbf{T}_{L,22}^{-1}\mathbf{T}_{I,22}^{-1}\mathbf{T}_{0,22}^{-1}\notag\\
&\overset{(a)}{=}\mathbf{H}_{RI}\mathbf{T}_{I,22}^{-1}\mathbf{H}_{IT}
\end{align}%
where ($a$) holds for $K\ne N$. Hence, \eqref{eq-channel-T} is proved.
\vspace{-2mm}
\bibliographystyle{IEEEtran}

\begin{thebibliography}{10}
\providecommand{\url}[1]{#1}
\csname url@samestyle\endcsname
\providecommand{\newblock}{\relax}
\providecommand{\bibinfo}[2]{#2}
\providecommand{\BIBentrySTDinterwordspacing}{\spaceskip=0pt\relax}
\providecommand{\BIBentryALTinterwordstretchfactor}{4}
\providecommand{\BIBentryALTinterwordspacing}{\spaceskip=\fontdimen2\font plus
\BIBentryALTinterwordstretchfactor\fontdimen3\font minus \fontdimen4\font\relax}
\providecommand{\BIBforeignlanguage}[2]{{%
\expandafter\ifx\csname l@#1\endcsname\relax
\typeout{** WARNING: IEEEtran.bst: No hyphenation pattern has been}%
\typeout{** loaded for the language `#1'. Using the pattern for}%
\typeout{** the default language instead.}%
\else
\language=\csname l@#1\endcsname
\fi
#2}}
\providecommand{\BIBdecl}{\relax}
\BIBdecl

\bibitem{An2023-ICC}
J.~An, M.~Di~Renzo, M.~Debbah, and C.~Yuen, ``Stacked intelligent metasurfaces for multiuser beamforming in the wave domain,'' in \emph{IEEE Int. Conf. Commun. (ICC)}, Rome, Italy, Jun. 2023, pp. 2834--2839.

\bibitem{An2024-WCM}
J.~An \emph{et~al.}, ``Stacked intelligent metasurface-aided {MIMO} transceiver design,'' \emph{IEEE Wireless Commun.}, vol.~31, no.~4, pp. 123--131, Aug. 2024.

\bibitem{Rev1_1}
H.~Liu \emph{et~al.}, ``Stacked intelligent metasurfaces for wireless sensing and communication: Applications and challenges,'' \emph{arXiv preprint arXiv:2407.03566}, 2024.

\bibitem{Boccia2012-TMTT}
L.~Boccia, I.~Russo, G.~Amendola, and G.~Di~Massa, ``Multilayer antenna-filter antenna for beam-steering transmit-array applications,'' \emph{IEEE Trans. Microw Theory Tech.}, vol.~60, no.~7, pp. 2287--2300, Jul. 2012.

\bibitem{Linglong2024-patent}
L.~Dai, L.~Kunzan, Z.~Zhang, R.~Mackenzie, and H.~Mo, ``Wireless telecommunications network including a multi-layer transmissive reconfigureable intelligent surface,'' Patent US 012\,028\,139B2, Jul. 2,, 2024.

\bibitem{Matteo2024-LCOM}
M.~Nerini and B.~Clerckx, ``Physically consistent modeling of stacked intelligent metasurfaces implemented with beyond diagonal {RIS},'' \emph{IEEE Commun. Lett.}, vol.~28, no.~7, pp. 1693--1697, Jul. 2024.

\bibitem{An2024-JSAC}
J.~An \emph{et~al.}, ``Two-dimensional direction-of-arrival estimation using stacked intelligent metasurfaces,'' \emph{IEEE J. Sel. Areas Commun}, vol.~42, no.~10, pp. 2786--2802, Oct. 2024.

\bibitem{abrardo2025novel}
A.~Abrardo, G.~Bartoli, and A.~Toccafondi, ``A novel comprehensive multiport network model for stacked intelligent metasurfaces ({SIM}) characterization and optimization,'' \emph{arXiv preprint arXiv:2501.02597}, 2025.

\bibitem{Markos2008wave}
P.~Markos and C.~M. Soukoulis, \emph{Wave propagation: From electrons to photonic crystals and left-handed materials}.\hskip 1em plus 0.5em minus 0.4em\relax Princeton University Press, 2008.

\bibitem{Frei2008-TMTT}
J.~Frei, X.-D. Cai, and S.~Muller, ``Multiport {$S$}-parameter and {$T$}-parameter conversion with symmetry extension,'' \emph{IEEE Trans. Microw Theory Tech.}, vol.~56, no.~11, pp. 2493--2504, Nov. 2008.

\bibitem{MatteoUniv2024-TWC}
M.~Nerini, S.~Shen, H.~Li, M.~Di~Renzo, and B.~Clerckx, ``A universal framework for multiport network analysis of reconfigurable intelligent surfaces,'' \emph{IEEE Trans. Wireless Commun.}, vol.~23, no.~10, pp. 14\,575--14\,590, Oct. 2024.

\bibitem{Hamad2024-OJCOM}
H.~Yahya, H.~Li, M.~Nerini, B.~Clerckx, and M.~Debbah, ``Beyond diagonal {RIS}: Passive maximum ratio transmission and interference nulling enabler,'' \emph{IEEE Open J Commun. Soc.}, vol.~5, pp. 7613--7627, Nov. 2024.

\bibitem{Gradoni2021-WCL}
G.~Gradoni and M.~Di~Renzo, ``End-to-end mutual coupling aware communication model for reconfigurable intelligent surfaces: An electromagnetic-compliant approach based on mutual impedances,'' \emph{IEEE Wireless Commun. Lett.}, vol.~10, no.~5, pp. 938--942, May 2021.

\bibitem{Ajam2022-TCOM}
H.~Ajam, M.~Najafi, V.~Jamali, B.~Schmauss, and R.~Schober, ``Modeling and design of {IRS}-assisted multilink {FSO} systems,'' \emph{IEEE Trans. Commun.}, vol.~70, no.~5, pp. 3333--3349, May 2022.

\bibitem{Rev2_1}
M.~Rezvani, R.~Adve, A.~b. Sediq, and A.~El-Keyi, ``Uplink wave-domain combiner for stacked intelligent metasurfaces accounting for hardware limitations,'' \emph{arXiv preprint arXiv:2407.21012}, 2024.

\end{thebibliography}

\end{document}